\documentclass[a4paper,11pt]{article}
\usepackage[utf8]{inputenc}
\usepackage{amsthm, amssymb, mathrsfs, amscd, amsmath, stmaryrd}
\usepackage{fancyhdr}
\usepackage{mathtools}
\usepackage{tikz}
\usepackage[pdfa]{hyperref}
\hypersetup{               
	colorlinks=true,                
	breaklinks=true,                
	urlcolor= black,                 
	bookmarksopen=false,
	filecolor=black,
	citecolor=black,
	linkbordercolor=blue
}
\usepackage{graphicx}
\usepackage{indentfirst} 
\usepackage{multirow}
\usepackage{latexsym}
\usepackage{longtable}
\usepackage{verbatim}
\usepackage{amsfonts}
\usepackage{graphics}
\usepackage{mathrsfs}
\usepackage{enumerate}
\usepackage{cite} 
\usepackage{url}
\usepackage{color,soul}
\usepackage{titling}
\usepackage{complexity}
\usepackage{cryptocode}
\usepackage{float}

\providecommand{\keywords}[1]
{
	\small	
	\textbf{\textit{Keywords---}} #1
}

\DeclarePairedDelimiter\absval{\lvert}{\rvert}
\newcommand{\abs}[1]{\absval*{#1}}

\newcommand{\qquote}[1]{``#1''}
\newcommand{\sep}{\,\,|\,\,}

\newlang{\GAIP}{GAIP}
\newlang{\mGAIP}{mGAIP}
\newlang{\pGAIP}{pGAIP}
\newlang{\dGAIP}{dGAIP}
\newlang{\Class}{C}

\theoremstyle{plain}
\newtheorem{Th}{Theorem}
\newtheorem{Prop}[Th]{Proposition}

\newtheorem{Cor}[Th]{Corollary}

\theoremstyle{definition}
\newtheorem{Def}[Th]{Definition}

\title{A Note on the Hardness of Problems from Cryptographic Group Actions}
\author{Giuseppe D'Alconzo\\ \href{mailto:giuseppe.dalconzo@polito.it}{giuseppe.dalconzo@polito.it} \\ Department of Mathematical Sciences, Politecnico di Torino}
\date{}

\fancyhfoffset[L]{1cm} 
\fancyhfoffset[R]{1cm} 
\begin{document}
	\maketitle
	
	\begin{abstract}
		Given a cryptographic group action, we show that the Group Action Inverse Problem ($\GAIP$) and other related problems cannot be $\NP$-hard unless the Polynomial Hierarchy collapses. We show this via random self-reductions and the design of interactive proofs. Since cryptographic group actions are the building block of many security protocols, this result serves both as an upper bound on the worst-case complexity of some cryptographic assumptions and as proof that the hardness in the worst and in the average case coincide. We also point out the link with Graph Isomorphism and other related $\NP$ intermediate problems.
	\end{abstract}
	
	
	\section{Introduction}
	\label{Sect:intro}
	
	\paragraph{Group Actions.}
	Cryptographic Group Actions (CGA) were introduced by Almati et al. \cite{alamati2020cryptographic} in 2020 and, under the name of Hard Homogeneous Spaces by Couveignes \cite{couveignes2006hard} in 2006. They are a powerful tool to design cryptographic protocols. The purpose of this setting is to generalise some well-known computational problems and assumptions like the Discrete Logarithm Problem (DLP). Their study had an increasing interest with the coming of Isogeny-based cryptography, in particular, among others, CSIDH \cite{castryck2018csidh} and Csi-fish \cite{beullens2019csi}, but more recently even with some code-based construction \cite{barenghi2021less} and lattices \cite{beullens2020calamari}.
	
	\paragraph{Computational Complexity Theory.}
	On the other hand, computational complexity has always been a tool in cryptographic proofs, serving as an upper bound on the hardness of certain problems. By its nature, cryptography has to deal with \emph{hard-on-average} problems and common assumptions embrace this setting. Classical complexity theory initially studied worst-case scenario until Levin, in its work \cite{levin1986average}, proposed a theory about problems with a distribution associated on their inputs. This defines a \emph{distributional problem} and the average-case study of worst-case (believed) intractable problems. A bridge between worst and average-case is built upon random reductions \cite{abadi1989hiding,feigenbaum1990locally,feigenbaum1993random} binding a problem $D$ to a distributional problem $D'$. In particular, in this work we will use \emph{(nonadaptively) random self-reductions} from \cite{feigenbaum1993random}, but an analogue construction leading the same results can be done with \emph{information hiding schemes} introduced in \cite{abadi1989hiding}. Another tool to study the worst-case hardness of decision problems is given by the so-called \emph{interactive proofs} (or protocols). In this setting, two machines, a Prover and a Verifier, exchange messages to decide whether an instance is a solution of a problem.
	
	\paragraph{This work.}
	We introduce some worst-case problems modelling computational assumptions related to group actions from \cite{alamati2020cryptographic}: by their nature, these assumptions cover the average-case view of these problems. We show that, if such problems are \NP-hard, then some widely believed complexity assumptions fall. In fact, we prove that $\GAIP$ and other problems from regular group actions are \emph{random self-reducible}, i.e. there exists a reduction from a particular instance to one or many random instances, and this, if they are \NP-hard, leads to the collapse of the Polynomial Hierarchy at the third level. This serves as an upper bound on the worst-case complexity of cryptographic group action and proves that the hardness in the average and the worst case is the same. We also study problems from non-transitive group actions and the strong link with the Graph Isomorphism Problem: in particular, the latter is an instance of the more general \emph{Decisional Group Inversion Problem} ($\dGAIP$). Other $\NP$ intermediate problems are shown to be particular instances of $\dGAIP$.  We prove that this problem is both in the class $\coAM$ and at the second level of the low hierarchy $\L_2^P$.\\
	This paper is organised as follows: in Section \ref{Sect:grAct} we recall the formal definition of group actions and their cryptographic properties (one-way-ness, weakly unpredictability and weakly pseudorandomness). In the successive Section \ref{Sect:CC} we introduce some tools from complexity theory: worst-to-average-case reductions and interactive protocols. Section \ref{Sect:ProblemsGA} concerns our main result, defining some computational problems from regular group actions and analysing them. We deal also with non-transitive actions, their connection with the Graph Isomorphism Problem and their belonging to some complexity classes. Then Section \ref{Sect:concl} contains conclusions and observations.
	
	\subsection{Cryptographic Group Actions}
\label{Sect:grAct}

\subsubsection{Effective Group Actions}
\label{Subsect:EGA}

We recall the definition of group action.
\begin{Def}
	Let $G$ be a multiplicative group and $X$ a set. We say that $G$ acts on $X$ if there exists a map $\star:G\times X \to X$ such that:
	\begin{enumerate}
		\item if $e$ is the identity of $G$, for every $x\in X$ we have $e\star x = x$ and
		\item for every $g_1,g_2 \in G$ and every $x\in X$, we have $(g_1g_2)\star x = g_1\star(g_2\star x)$.
	\end{enumerate}
	The group action is denoted with $(G,X,\star)$.
\end{Def}
A group action  $(G,X,\star)$ can satisfy some particular properties, it can be:
\begin{itemize}
	\item \emph{Transitive}: for every $x_1,x_2\in X$ there exists $g\in G$ such that $x_1=g\star x_2$.
	\item \emph{Free}: for every $g\in G$, $g$ is the identity if and only if there exists $x\in X$ such that $g\star x = x$.
	\item \emph{Regular}: if $(G,X,\star)$ is both free and transitive.
\end{itemize}
We can see that if the action $(G, X, \star)$ is regular and the group $G$ is finite, then for every $x\in X$ the map $f_x:g\mapsto g\star x$ is a bijection and $\abs{G}=\abs{X}$. Given a transitive action, we can define $\delta(x,y)$ as the element of $G$ for which $x=\delta(x,y) \star y$. If the action is non-transitive, the element $\delta(x,y)$ could not exist: we will denote this fact saying that $\delta(x,y)$ is not in $G$.

To be suitable in cryptography, \cite{alamati2020cryptographic} defines the concept of \emph{effective} group action.
\begin{Def}\label{Def:Eff}
	A group action $(G,X,\star)$ is \emph{effective} if
	\begin{itemize}
		\item The group $G$ is finite and there exist probabilistic polynomial-time (PPT) algorithms for
		\begin{enumerate}
			\item \emph{Membership testing}: decide whether a bit-string represents an element of $G$.
			\item \emph{Equality testing}: given two bit-strings, decide whether they represent the same element of $G$.
			\item \emph{Sampling}: given a distribution $\mathcal{D}_G$ on $G$, sample with respect to $\mathcal{D}_G$.
			\item \emph{Operation}: compute $g_1g_2$ for every $g_1,g_2\in G$.
			\item \emph{Inversion}: compute $g^{-1}$ for every $g\in G$.
		\end{enumerate}
		\item The group $X$ is finite and there exist PPT algorithms for
		\begin{enumerate}
			\item \emph{Membership testing}: decide whether a bit-string represents an element in $X$.
			\item \emph{Unique representation}: given an element in $X$ compute a bit-string that canonically represent it.
		\end{enumerate}
		\item There is an efficient algorithm that given $g\in G$ and $x\in X$ computes $g\star x$.
	\end{itemize}
\end{Def}
Another optional property is the existence of a particular set element $x_0\in X$ called \emph{origin}. Given a free action $(G,X,\star)$, we can construct the regular action $(G,X_0,\star)$, where $X_0=\{ g\star x_0 \sep g\in G \}$. \\
For simplicity, since we can always construct a regular action from a free one, we will focus on them in Subsection \ref{Subsect:RegAct} and on non-transitive ones in Subsection \ref{Subsect:non-tr}. Note that in \cite{couveignes2006hard}, Hard Homogeneous Spaces are defined from a regular group action.

There is a particular case when we do not assume the Unique representation property for $X$. In fact, we admit that computing a canonical form for every element $x$ in $X$ can be hard. An example of such group action can be seen in Subsection \ref{Subsect:non-tr}.

\subsubsection{Cryptographic Assumptions}
\label{Subsect:CryAss}
Observe that all the requirements in the previous definitions lead to an efficient (and then, effective) use of the action in a protocol. In order to use them in cryptography we need some computational assumptions, as defined in \cite{alamati2020cryptographic}.\\
With $\mathbf{P}[A]$ we denote the probability of the event $A$. A function $f(\lambda)$ is \emph{negligible} in $\lambda$ if there exists an $\lambda_0$ such that for every $\lambda\ge \lambda_0$ we have $f(\lambda)< \frac{1}{\lambda^c}$ for every $c\in\mathbb{N}$.

In \cite{alamati2020cryptographic}, they define concepts of one-way function and weakly unpredictable and weakly pseudorandom permutations. To ease the notation, we directly define these concepts in terms of group actions.
\begin{Def}
	Let $\lambda$ be a parameter indexing $G$ and $X$. Given $\mathcal{D}_G$ and $\mathcal{D}_X$ be two distributions over $G$ and $X$ respectively, then the group action $(G,X,\star)$ is \emph{$(\mathcal{D}_G,\mathcal{D}_X)$-one-way} if, given the family $\{f_x:G\to X\}_{x\in X}$, for all PPT adversaries $\mathcal{A}$ the function
	$$ P(\lambda) := \mathbf{P}[f_x\left(\mathcal{A}\left(x,f_x(g)\right)\right) = f_x(g)] $$
	is negligible in the parameter $\lambda$,
	where $x$ is sampled according to $\mathcal{D}_X$ and $g$ according to $\mathcal{D}_G$.
\end{Def}

For any $g\in G$ we define the map $\pi_g:X\to X$ where $\pi_g(x)=g\star x$. If the action is regular, $\pi_g$ is a permutation.

\begin{Def}
	Let $\lambda$ be a parameter indexing $G$ and $X$. Given $\mathcal{D}_G$ and $\mathcal{D}_X$ be two distributions over $G$ and $X$ respectively, then the group action $(G,X,\star)$ is \emph{$(\mathcal{D}_G,\mathcal{D}_X)$-weakly unpredictable} if, given the family of permutations $\{\pi_g:X\to X\}_{g\in G}$ and the randomised oracle $\Pi_g$ that, when queried, samples $x$ from $\mathcal{D}_X$ and outputs $(x,\pi_g(x))$, for all PPT adversaries $\mathcal{A}$ the function
	$$ P(\lambda) := \mathbf{P}[\mathcal{A}^{\Pi_g}\left(y\right) = \pi_g(y)] $$
	is negligible in the parameter $\lambda$,
	where $y$ is sampled according to $\mathcal{D}_X$ and $g$ according to $\mathcal{D}_G$.
\end{Def}

With $\mathcal{S}_X$ we denote the set of permutations over the set $X$.

\begin{Def}
	Let $\lambda$ be a parameter indexing $G$ and $X$. Given $\mathcal{D}_G$ and $\mathcal{D}_X$ be two distributions over $G$ and $X$ respectively, then the group action $(G,X,\star)$ is \emph{$(\mathcal{D}_G,\mathcal{D}_X)$-weakly pseudorandom} if, given the family of permutations $\{\pi_g:X\to X\}_{g\in G}$ and two randomised oracles
	\begin{enumerate}
		\item $\Pi_g$ that, when queried samples $x$ from $\mathcal{D}_X$, outputs $(x,\pi_g(x))$;
		\item $U$ that, when queried samples $x$ from $\mathcal{D}_X$ and $\sigma\in \mathcal{S}_X$ uniformly at random, outputs $(x,\sigma(x))$;
	\end{enumerate}
	for all PPT adversaries $\mathcal{A}$ the function
	$$ P(\lambda) := \abs{\mathbf{P}[\mathcal{A}^{\Pi_g}\left(1^\lambda\right) = 1] - \mathbf{P}[\mathcal{A}^{U}\left(1^\lambda\right) = 1] }$$
	is negligible in the parameter $\lambda$,
	where $y$ is sampled according to $\mathcal{D}_X$ and $g$ according to $\mathcal{D}_G$.
\end{Def}

The flavour of the previous computational definitions is mostly cryptographic but we want to investigate the hardness of problems they implicate. In Section \ref{Sect:ProblemsGA} we will focus on worst-case problems directly related to these computational assumptions.

Cryptographic Group Actions, i.e. group actions for which the above properties are assumed, can be used as a tool in the design of a variety of different cryptographic protocols: key exchange, sigma protocols, dual-mode public key encryption, two-message statistically sender-private oblivious transfer and Naor-Reinold pseudorandom function. In many of these applications is required an abelian group, however in this work we consider the generic case. An overview is given in \cite{alamati2020cryptographic}.

\subsection{Tools from Computational Complexity}
\label{Sect:CC}

Let $\{0,1\}^*$ be the set of binary string of finite length and $\{0,1\}^\lambda$ the set of binary string of length $\lambda$. For any $x\in\{0,1\}^*$, we denote with $\abs{x}$ its length.

\subsubsection{Polynomial Hierarchy and Low Sets}
\label{Subsect:PH}
A \emph{decision problem} is a set $D$ of binary strings. We ask if, given a string $x$, it belongs to $D$. A \emph{search problem} $S$ is a set of couples of binary strings $(x,y)$ and we ask, given $x$, to find $y$ such that $(x,y)\in S$.\\
The class $\P$ consists of decision problems that can be solved in (deterministic) polynomial-time, while the class $\NP$ contains decision problems solvable by a nondeterministic polynomial-time algorithm (or, equivalently, by a nondeterministic Turing Machine). Alternatively, $\NP$ can be defined as all the problems $D$ such that for every instance $x$, if $x$ is in $D$, then there exists a \qquote{witness} $y_x$ such that $x\in D$ can be verified in polynomial time thanks to $y_x$. A problem is said \emph{$\NP$-hard} if every problem in $\NP$ can be reduced to it in polynomial time. More generally, given a class $\Class$, we say that a problem is \emph{$\Class$-hard} if every problem in that class can be reduced to it. A problem is \emph{$\Class$-complete} if it $\Class$-hard and it is in $\Class$.\\
The \emph{complement} of a decision problem $D$ is denoted with $\overline{D}$ and consists of all the strings $x$ that are not in $D$. Given a class $\Class$, we define the class $\co \Class$, containing all the complements of problems in $\Class$.\\
The \emph{Polynomial Hierarchy} ($\PH$) \cite{stockmeyer1976polynomial} is a chain of classes $\Sigma_{i}^P$ and $\Pi_i^P$ indexed by non-negative integers $i$, where $\Sigma_{i}^P=\co\Pi_i^P$. At lower levels we have $\Sigma_0^P=\Pi_0^P=\textsc{P}$ and $\Sigma_{1}^P=\NP$. We define $\PH=\bigcup_{i\ge 0} \Sigma_i^P$, or equivalently $\PH=\bigcup_{i\ge 0} \Pi_i^P$. It is known that 
each level $i$ is contained in the $(i+1)$-th one, but it is not known if these inclusions are proper. If there exists an $i$ such that $\Sigma_{i}^P=\Sigma_{i+1}^P$, we have that $\PH=\Sigma_i^P$ \cite{stockmeyer1976polynomial} and we say that the Polynomial Hierarchy \emph{collapses at the $i$-th level}. It is believed that $\PH$ does not collapses at any level.

We can construct another hierarchy called \emph{Low Hierarchy} \cite{schoning1983low}. Let $\P(D)$ be the set of decision problems solved by a polynomial-time machine with access to the oracle solving the problem $D$. Analogously let $\NP(D)$ be the set of decision problems solved by a non-deterministic polynomial-time machine with access to the oracle solving the problem $D$. We can extend the definitions above to oracles solving a class $\Class$ of problems taking the union $\NP(\Class)=\bigcup_{D\in \Class}\NP(D)$. Now define $\Sigma_0(\Class)=\P(\Class)$ and $\Sigma_{k}^P(\Class)=\NP(\Sigma_{k-1}^P(\Class))$. A decision problem $D$ is said \emph{low for the level $k$}, in symbols $D\in \L^P_k$, if $\Sigma_{k}^P(D)=\Sigma_{k}^P$.

\subsubsection{Random Self-Reductions}
\label{Subsect:rsr}
We observe that in \cite{feigenbaum1993random} most of the definitions are based on decision problems. In this work, we deal both with decision and search problems, and the core concept of \emph{random self-reduction}, introduced in \cite{feigenbaum1993random}, is adapted to our case.

For a binary string $x$, set $|x|=\lambda$ and with $\mathfrak{r}$ we denote the randomness used in the following reduction. The length of $\mathfrak{r}$ is a polynomially bounded function $\omega$ of $\lambda$, that is $|\mathfrak{r}|=\omega(\lambda)$. Also $k=k(\lambda)$ is a polynomially bounded function.
\begin{Def}
	\label{Def:srs}
	Let $L$ be a search or decision problem. If $x$ is an instance for $L$, we denote with $F_L(x)$ the function solving $L$ on $x$. The problem $L$ is said \emph{nonadaptively $k$-random self-reducible} if there are polynomial computable functions $\sigma$ and $\phi$ such that
	\begin{itemize}
		\item[-] for any $\lambda$ and $x\in\{0,1\}^\lambda,$ the reduction is correct with a good probability:
		$$F_L(x) = \phi(x,\mathfrak{r},F_L(\sigma(1,x,\mathfrak{r})), \dots, F_L(\sigma(k,x,\mathfrak{r}))$$
		for at least $\frac34$ of all $\mathfrak{r}$ in $\{0,1\}^{\omega(\lambda)}$;
		\item[-] for any $\lambda$ and any pair of instances $x_1,x_2\in\{0,1\}^\lambda$, if $\mathfrak{r}$ is chosen uniformly at random, then $\sigma(i,x_1,\mathfrak{r})$ and $\sigma(i,x_2,\mathfrak{r})$ are instances of $L$ for every $1\le i\le k$. Moreover, they are identically distributed for every $1\le i\le k$.
	\end{itemize}
\end{Def}
We briefly comment this definition: on input the round $i$ and the instance $x$, $\sigma$ generates another instance based on the randomness given. These new instances and the randomness $\mathfrak{r}$ are given to the reduction $\phi$ that, using an oracle on these random instaces and the initial one $x$, returns the answer for $x$.

The following result is proven in \cite{feigenbaum1993random} for decision problem.
\begin{Th}
	\label{Th:rsr}
	If a decision problem $L$ is complete for $\Sigma_i^P$ or $\Pi_i^P$ and it is nonadaptively $k$-random self-reducible, then the Polynomial Hierarchy collapses at level $i+2$, that is $\PH =\Sigma_{i+2}^P$.
\end{Th}

The class $\NP$ corresponds to $\Sigma_1^P$ and we have the following corollary, originally stated in \cite{feigenbaum1993random} for decision problem, but adapted in \cite{feigenbaum1990locally} for generic \NP-hard problems.
\begin{Cor}\label{Cor:PH}
	If an \NP-hard problem is nonadaptively $k$-random-self-reducible, then the Polynomial Hierarchy collapses at the third level: $\PH =\Sigma_{3}^P$.
\end{Cor}

Since it is widely believed that the Polynomial Hierarchy does not collapse at any level, no complete problems that are nonadaptively $k$-random-self-reducible should exist.

\subsubsection{Interactive Proofs and Arthur-Merlin Games}
\label{Subsect:IP-AM}
Now we deal with interactive protocols and proofs. We recall the definition of the class $\IP$, introduced in \cite{goldwasser1989knowledge}.
\begin{Def}\label{Def:IP}
	Let $f$ be an integer function. A decision problem $D$ is in $\IP[f(n)]$ if there exist two algorithms, \textsc{Prover} and \textsc{Verifier}, such that on common input $x$:
	\begin{enumerate}
		\item \textsc{Prover} and \textsc{Verifier} exchange $f(\abs{x})$ messages;
		\item \textsc{Verifier} is probabilistic and polynomial-time in $\abs{x}$;
		\item \textsc{Prover} is computationally unbounded;
		\item if $x$ is in $D$, then \textsc{Verifier} accepts with a probability of at least $\frac34$;
		\item if $x$ is not in $D$, then \textsc{Verifier} accepts with probability at most $\frac14$.
	\end{enumerate}
	The class $\IP$ is defined as $\IP=\bigcup_{k\ge 0}\IP[n^k]$.
\end{Def}
We can see that $\NP$ is contained in $\IP[1]$, interactive proofs with only one round of communication.\\
Another useful class consists in Arthur-Merlin games (or protocols) \cite{babai1988arthur}. In brief, an $\AM$ protocol is an interactive protocol where \textsc{Prover} sends only random strings.
We define $\AM=\AM[2]$ and the reason will be given in the next theorem. The main difference between $\IP$ and $\AM$ is that in the latter, the randomness produced by \textsc{Verifier} is public, while for $\IP$ this is not needed. The next result synthesises some facts about $\IP$ and $\AM$ \cite{goldwasser1986private,babai1988arthur}.
\begin{Th}\label{Th:IP-AM}
	For every integer function $f$:
	\begin{itemize}
		\item $\AM[k] = \AM $ for every $k \ge 2$;
		\item $\IP[f(n)]\subseteq \AM[f(n)+2]$ for every positive integer $n$.
	\end{itemize}
\end{Th}

Now we present results from \cite{boppana1987does,schoning1988graph}, dealing with the collapse of $\PH$.
\begin{Th}\label{Th:collapse}
	We have that $\NP\cap \coAM \subseteq \L_2^P$. Moreover, if $\NP\subseteq \coAM$, then $\PH$ collapses to the second level.
\end{Th}
In the following sections, we will show that some problems cannot be $\NP$-hard unless $\PH=\Sigma_{2}^P$ or $\PH=\Sigma_{3}^P$.

\section{Problems from Group Actions}
\label{Sect:ProblemsGA}

We consider group actions $(G,X,\star)$ where the length of the bit-string representation of $X$ (or $G$) is polynomial in the length of the representation of $G$ (or $X$, respectively). Moreover, in this section we assume that, in the definition of effective group action, all the operations involving $G$ and $X$ can be performed in polynomial time (instead of probabilistic polynomial time, see Definition \ref{Def:Eff}).

\subsection{Regular Group Actions}
\label{Subsect:RegAct}

From now on, we refer to $\lambda$ as the complexity parameter indexing $G$ and $X$. For instance, $\lambda$ can be the length of the bit-string representation of $G$ and $X$. We define the following computational problems.
\begin{Def}
	Let $(X,G,\star)$ be a group action. Given $x$ and $y$ in $X$, find, if there exists, an element $g\in G$ such that $x=g\star y$. This problem is called \emph{Group Action Inverse Problem} ($\GAIP$).
\end{Def}
We can see that the $\GAIP$ models in the worst-case the cryptographic concept of one-way group action. In fact, if there exists a polynomial-time algorithm that solves $\GAIP$, then it will break the One-way group action assumption.

\begin{Def}
	Let $(X,G,\star)$ be a group action and let $q$ be an integer polynomial. Given the set of pairs $\{(x_i,y_i) \}_{i=1}^{q(\lambda)}$, find, if there exists, an element $g\in G$ such that $x_i=g\star y_i$ for every $i$. This problem is called \emph{Multiple Group Action Inverse Problem} ($\mGAIP$).
\end{Def}

Even in this case, $\mGAIP$ models the worst-case scenario of a $(\mathcal{D}_G,\mathcal{D}_X)$-weakly unpredictable group action.

\begin{Def}
	Let $(X,G,\star)$ be a group action and let $q$ be an integer polynomial. Let $L_{g}$ be the set indexed by $g\in G$ containing all the pairs $(x,g\star x)$. Given the set of pairs $\{(x_i,y_i) \}_{i=1}^{q(\lambda)}$, decide if it is a subset of $L_{g}$ for some $g$ or it is sampled uniformly from $X\times X$. This problem is called \emph{Pseudorandom Group Action Inverse Problem} ($\pGAIP$).
\end{Def}

Finally, $\pGAIP$ is the worst-case scenario of a $(\mathcal{D}_G,\mathcal{D}_X)$-weakly pseudorandom group action. Observe that both $\mGAIP$ and $\pGAIP$ are polynomially reducible to $\GAIP$.

\begin{Th}\label{Th:problemsRsr}
	Let $(G,X,\star)$ be a regular effective group action. Then
	\begin{enumerate}
		\item Group Action Inverse Problem,
		\item Multiple Group Action Inverse Problem,
		\item Pseudorandom Group Action Inverse Problem
	\end{enumerate}
	are nonadaptively 1-random self-reducible.
\end{Th}
\begin{proof}
	\begin{enumerate}
		\item We denote with $F_{\GAIP}$ the function that solves the $\GAIP$  problem: on input $(x,y)$, $F_{\GAIP}(x,y)$ is an element of $G$ such that $x = F_{\GAIP}(x,y)\star y$. We show that there exist polynomial computable functions $\sigma$ and $\phi$ that satisfy Definition \ref{Def:srs}. \\
		Let $x,y\in X$, since the action is regular, there exists $g$ for which $g\star x=y$. Let $\mathfrak{r}$ be the randomness string of length polynomial in $\lambda=\abs{x}$. Set $\sigma$ as the function whose, on inputs $(i,(x,y),\mathfrak{r})$, deterministically generates elements $g_{x,i}$ and $g_{y,i}$ from $\mathfrak{r}$ and $i$, and outputs $(g_{x,i}\star x, g_{y,i}\star y)$. By construction $\sigma(i,(x,y),\mathfrak{r})$ is an instance of $\GAIP$. Then define the function $\phi$ that on inputs $((x,y),\mathfrak{r}, F_\GAIP(\sigma(1,(x,y),\mathfrak{r})))$, generates $g_{x,1}$ and $g_{y,1}$ from $\mathfrak{r}$, and returns $g_{x,1}^{-1} \cdot F_\GAIP(\sigma(1,(x,y),\mathfrak{r})) \cdot g_{y,1}$. If we denote with $x_1,y_1$ the output of $\sigma(1,(x,y),\mathfrak{r})$, we can verify
		\begin{align*}
			x &= g_{x,1}^{-1} \star x_1  \\
			& = g_{x,1}^{-1} \star (F_\GAIP(x_1,y_1) \star y_1) = (g_{x,1}^{-1} \cdot F_\GAIP(x_1,y_1)) \star (g_{y,1} \star y)  \\
			&= (g_{x,1}^{-1} \cdot F_\GAIP(x_1,y_1) \star \cdot g_{y,1}) \star y
		\end{align*}
		and the reduction is correct with probability 1. By construction, if $\mathfrak{r}$ is chosen uniformly at random, then the random variables $(\tilde{x},\tilde{y})=\sigma(1,(x,y),\mathfrak{r})$ and $(\tilde{x}',\tilde{y}')=\sigma(1,(x',y'),\mathfrak{r})$ are identically distributed. Indeed let $z,w\in X$, then
		\begin{align*}
			\mathbf{P}[\tilde{x}=z,\tilde{y}=w] &= \mathbf{P}[g_{x,i}=\delta(z,x),g_{y,i}=\delta(w,y)]\\
			& = \frac{1}{\abs{G}^2} = \mathbf{P}[g_{x',i}=\delta(z,x'),g_{y',i}=\delta(w,y')] \\
			&= \mathbf{P}[\tilde{x}'=z,\tilde{y}'=w],
		\end{align*}
		where the second and fourth equalities use the fact that $(g_{x,i},g_{y,i})$ and $(g_{x',i},g_{y',i})$ are picked independently from $(x,y)$ and $(x',y')$.
		This implies that $\GAIP$ is nonadaptively 1-random self-reducible.
		
		\item Suppose $F_{\mGAIP}$ is the function that solves $\mGAIP$. This reduction is very similar to the one for $\GAIP$. Given $Q=\{(x_j,y_j) \}_{j=1}^{q(\lambda)}$, where $q$ is a polynomial, on input $(i,Q,\mathfrak{r})$, $\sigma$ generates $g_{Q,i}$ and sets
		$\tilde{x}_j=g_{Q,i} \star x_j$ and $\tilde{y}_j=g_{Q,i} \star y_j$
		for each $j=1,\dots, q(\lambda)$, then outputs $\{\tilde{x}_j, \tilde{y}_j\}_{j=1}^{q(\lambda)}$. By construction $\sigma(i,Q,\mathfrak{r})$ is an instance of $\mGAIP$. The function $\phi$, on inputs $(Q,\mathfrak{r}, F_\mGAIP(\sigma(1,Q,\mathfrak{r})))$, retrieves $g_{Q,1}$ and outputs $g_{Q,1}^{-1}\cdot F_\mGAIP(\sigma(1,Q,\mathfrak{r}))) \cdot g_{Q,1}$. Let $\tilde{Q}=\sigma(1,Q,\mathfrak{r})=\{\tilde{x}_j, \tilde{y}_j\}_{j=1}^{q(\lambda)}$, then
		\begin{align*}
			(g_{Q,1}^{-1}\cdot F_\mGAIP(\tilde{Q})) \cdot g_{Q,1}) \star y_i &= (g_{Q,1}^{-1}\cdot F_\mGAIP(\tilde{Q})) \star (g_{Q,1} \star y_j)  \\
			& = g_{Q,1}^{-1}\star \left( F_\mGAIP(\tilde{Q})) \star (g_{Q,1} \star y_j)\right) \\
			&= g_{Q,1}^{-1}\star (g_{Q,1}\star x_j) = x_j
		\end{align*}
		for every $j=1,\dots,q(\lambda)$.
		Given $Q_1$ and $Q_2$ of the same form of $Q$ above, if $\mathfrak{r}$ is chosen uniformly random, then $\sigma(i,Q_1,\mathfrak{r})$ and $\sigma(i,Q_2,\mathfrak{r})$ are identically distributed: let $z,w\in X$, then for every $(x_j,y_j)$ in $Q_1$ and $(x'_j,y'_j)$ in $Q_2$ we have
		\begin{equation}\label{eq:unif1}
			\begin{aligned}
				\mathbf{P}[\tilde{x}_j = z,\tilde{y}_j = w] 	&=\mathbf{P}[g_{Q_1,i} =\delta(z,x_j),g_{Q_1,i} =\delta(w,y_j)]  \\
				& =\mathbf{P} 	[g_{Q_2,i}=\delta(z,x'_j),g_{Q_2,i} =\delta(w_j,y'_j)] \\
				& = \mathbf{P}[\tilde{x}'_j= z,\tilde{y}'_j = 	w] ,
			\end{aligned}
		\end{equation}
		where we identify with $\{\tilde{x}_j, \tilde{y}\}_{j=1}^{q(\lambda)}$ and $\{\tilde{x}'_j,\tilde{y}'\}_{j=1}^{q(\lambda)}$ the output of $\sigma(i,Q_1,\mathfrak{r})$ and $\sigma(i,Q_2,\mathfrak{r})$, respectively.\\			
		We conclude that $\mGAIP$ is nonadaptively 1-random self-reducible.
		
		\item Suppose $F_{\pGAIP}$ is the function that decides $\pGAIP$. Given $Q=\{(x_j,y_j)\}_{j=1}^{q(\lambda)}$, we have $F_{\pGAIP}(Q)=1$ if there exists $g\in G$ such that $x_j=g\star y_j$ for every $j=1,\dots,q(\lambda)$, while $F_{\pGAIP}(Q)=0$ if $x_j$ and $y_j$ are sampled from the uniform distribution over $X$.\\
		Given the instance $Q=\{(x_j,y_j) \}_{j=1}^{q(\lambda)}$, on input $(i,Q,\mathfrak{r})$, $\sigma$ generates $g_{Q,i}$ and sets
		$\tilde{x}_j=g_{Q,i} \star x_j$ and $\tilde{y}_j=g_{Q,i} \star y_j$
		for each $j=1,\dots, q(\lambda)$, then outputs $\{\tilde{x}_j, \tilde{y}_j\}_{j=1}^{q(\lambda)}$. We can see that $\sigma(i,Q,\mathfrak{r})$ is an instance of $\pGAIP$: if $F_\pGAIP(Q)=1$, then it fits in the definition. Otherwise when $F_\pGAIP(Q)=0$, then the output of $\sigma$ is a set of uniformly picked couple of elements in $X$. Let $z,w\in X$, then since $\pi_{g_{Q,i}}$ is a permutation,
		
		\begin{equation}\label{eq:unif2}
			\begin{aligned}
				\mathbf{P}[\tilde{x}_j = z,\tilde{y}_j = w] &= \mathbf{P}[g_{Q,i}\star x_j =z,g_{Q,i}\star y_j =w] \\
				&= \mathbf{P}[x_j = z, y_j=w] = \frac{1}{\abs{X}}\cdot\frac{1}{\abs{X}}.
			\end{aligned}
		\end{equation}
		The function $\phi$, on inputs $(Q,\mathfrak{r}, F_\pGAIP(\sigma(1,Q,\mathfrak{r})))$, simply outputs \\$F_\pGAIP(\sigma(1,Q,\mathfrak{r}))$. We prove that the reduction is correct: let 
		$$\tilde{Q}=\sigma(1,Q,\mathfrak{r})=\{\tilde{x}_j, \tilde{y}_j\}_{j=1}^{q(\lambda)},$$
		then
		\begin{itemize}
			\item if $F_\pGAIP(\tilde{Q})=1$, then there exists $\tilde{g}$ such that $\tilde{x}_j=\tilde{g}\star \tilde{y}_j$ for every $j$. Then we have that $x_j = (g_Q^{-1}\cdot \tilde{g} \cdot g_Q) \star y_j$ and this implies $F_\pGAIP(Q)=1$.
			\item If $F_\pGAIP(\tilde{Q})=0$, then $\tilde{x}_j, \tilde{y}_j$ are sampled from the uniform distribution over $X$, and so are $g_Q^{-1}\star \tilde{x}_j = x_j$ and $g_Q^{-1}\star \tilde{y}_j = y_j$. This implies $F_\pGAIP(Q)=0$.
		\end{itemize}
		Given $Q_1$ and $Q_2$ of the same form of $Q$ above, if $\mathfrak{r}$ is chosen uniformly random, then $\sigma(1,Q_1,\mathfrak{r})$ and $\sigma(1,Q_2,\mathfrak{r})$ are identically distributed due to equations \eqref{eq:unif1} and \eqref{eq:unif2}. Then $\pGAIP$ is nonadaptively 1-random self-reducible.
	\end{enumerate}
	
\end{proof}

We point out that these problems fit also in the definition of \emph{information hiding schemes} (ihs) from \cite{abadi1989hiding}, in fact, both $\GAIP$, $\mGAIP$ and $\pGAIP$ are 1-ihs using the proof of Theorem \ref{Th:problemsRsr}.\\
Due to the previous result and to Corollary \ref{Cor:PH}, we can show the following fact.

\begin{Cor}
	If $\GAIP$, $\mGAIP$ or $\pGAIP$ are \NP-hard, the Polynomial Hierarchy collapses at the third level.
\end{Cor}

\subsection{Non-transitive Group Actions}
\label{Subsect:non-tr}

Now we consider non-transitive group actions where, given $x,y$ in $X$, we do not know if there exists $\delta(x,y)$. We want to analyse the complexity of a decisional variant of $\GAIP$, using the same technique of \cite{petrank1997code} for Code Nonequivalence Problem and \cite{goldreich1991proofs} for Graph Nonisomorphism Problem.

\begin{Def}
	Let $(G,X,\star)$ be an effective non-transitive group action. Given $x$ and $y$ in $X$ decide whether $\delta(x,y)$ exists. This problem is called \emph{Decisional Group Inversion Problem} ($\dGAIP$).
\end{Def}

It is easy to see that $\dGAIP$ is in $\NP$: if $(x,y)$ is a yes-instance, the element $\delta(x,y)$ is a witness.\\
This problem embraces many decision problems present in literature.
\begin{enumerate}
	
	\item Let $G=\mathcal{S}_{\{1,\dots,n\}}$ be the set of permutations of $n$ elements, and let $X=\mathcal{G}_n$ be the set of graphs with $n$ vertices. The action of $G$ over $X$ is defined as the permutation of the vertices of the graph. The $\dGAIP$ for the above group action is the \emph{Graph Isomorphism Problem}. Many other isomorphism problems fall in this framework, for example the \emph{Ring Isomorphism Problem} \cite{kayal2005ring} and the \emph{Polynomial Isomorphism Problem} \cite{faugere2006polynomial}.
	
	\item Let $\mathbb{F}_2$ be the field with two elements. Set $G=\textsc{GL}_k\left(\mathbb{F}_2\right)\times \mathcal{S}_{\{1,\dots,n\}}$, where $\textsc{GL}_k\left(\mathbb{F}_2\right)$ is the set of $k\times k$ invertible matrices over the field of two elements and $\mathcal{S}_{\{1\dots,n\}}$ be the set of permutation of $n$ elements represented by binary $n\times n$ matrices. Let $X=\mathcal{C}_{n,k}\left(\mathbb{F}_2\right)$ be the set of generator matrices of linear codes over $\mathbb{F}_2$ having length $n$ and dimension $k$. The action of $G$ over $X$ is defined as
	$$
	\begin{array}{rccl}
		\star:&G\times X&\to&X\\
		&\left((S,P),A\right)&\mapsto & (S,P)\star A = SAP.
	\end{array}
	$$
	The $\dGAIP$ for the above group action coincides with the \emph{Permutation Code Equivalence Problem}.

	\item If $\mathbf{G}$ is a graph with $n$ vertices and $\pi\in \mathcal{S}_{\{1,\dots,n\}}$, we denote with $\pi\left(\mathbf{G}\right)$ the graph with vertices permutated by $\pi$. Given a graph $\mathbf{G}$ with vertices $\{1,\dots,n\}$, the graph $\mathbf{G}^{(i)}$ is the subgraph with vertices $\{1,\dots,n\}\setminus \{i\}$. A \emph{deck} of $\mathbf{G}$ is a $n$-uple $D(\mathbf{G})=\left(\mathbf{G}^{(1)},\dots,\mathbf{G}^{(n)}\right)$ and we say that $\mathbf{G}$ is a \emph{reconstruction} of $D(\mathbf{G})$. Given a graph $\mathbf{G}$ and a sequence of graphs $D=\left(\mathbf{G}_1,\dots,\mathbf{G}_n\right)$, the problem of deciding whether $\mathbf{G}$ is a reconstruction of $D$ is called \emph{Deck Checking Problem} \cite{kobler1992graph,kratsch1994complexity}. Now let $G=\left(\mathcal{S}_{\{1,\dots,n\}}\right)^{n+1}$ and $X=\left(\mathcal{G}_{n-1}\right)^n$. Define $\star$ as follows
	$$
	\begin{array}{rccc}
		\star:&G\times X&\to&X\\
		&\left( (\sigma,\pi_1,\dots,\pi_n) , \left(\mathbf{G}_1,\dots,\mathbf{G}_n\right)\right)&\mapsto & \left(\pi_1\left(\mathbf{G}_{\sigma(1)}\right),\dots,\pi_n\left(\mathbf{G}_{\sigma(n)}\right)\right).
	\end{array}
	$$
	We can see the Deck Checking Problem as the instance $\left(D(\mathbf{G}),D\right)$ of $\dGAIP$ for the previous group action. This is not surprising since Deck Checking can be reduced to Graph Isomorphism.
		
	\item Now we admit that the canonical form for $X$ can be computed having access to an oracle for $\NP$. Let $G=\mathcal{S}_{\{1,\dots,n\}}$ and $X$ be the set of boolean functions on $n$ variables, the action is defined as the permutation of variables $x_1,\dots,x_n$. The $\dGAIP$ for this group action is called \emph{Boolean Isomorphism Problem} and it is in $\Sigma_{2}^P$ and $\coNP$-hard \cite{borchert1998computational}, but unlikely to be $\Sigma_{2}^P$-complete, unless $\PH=\Sigma_{3}^P$ \cite{agrawal1996boolean}.
\end{enumerate}

Due to the similarity to such known problems, we can prove the following theorem.

\begin{Prop}\label{Prop:dGAIPcoAM}
	Let $(G,X,\star)$ be an effective non-transitive group action. Then $\dGAIP$ is in $\coAM$.
\end{Prop}
\begin{proof}
	This proof uses the same technique used in Section III of \cite{petrank1997code}. We design an interactive protocol for $\overline{\dGAIP}$, the complement of $\dGAIP$.
	
	We can show that the protocol in Figure \ref{Fig:dGAIP} satisfies requirements in Definition \ref{Def:IP}. Given a couple $y_0,y_1$ such that they belong to $\overline{\dGAIP}$, \textsc{Verifier} will always accepts. Conversely, if $(y_0,y_1)$ does not belong to $\overline{\dGAIP}$, then there exists $h$ such that $y_0=h\star y_1$ and then \textsc{Prover} can only guess $b'$ since both $\delta(x,y_0)$ and $\delta(x,y_1)$ are in $G$. Then $b=b'$ happens with probability $\frac12$. Repeating the protocol a polynomial number of time gives us a probability at most $\frac14$ that \textsc{Verifier} accepts. Then $\overline{\dGAIP}$ is in $\IP[2]\subseteq \AM[4] \subseteq \AM$. This implies that $\dGAIP$ is in $\coAM$. 
	
\end{proof}

\begin{figure}
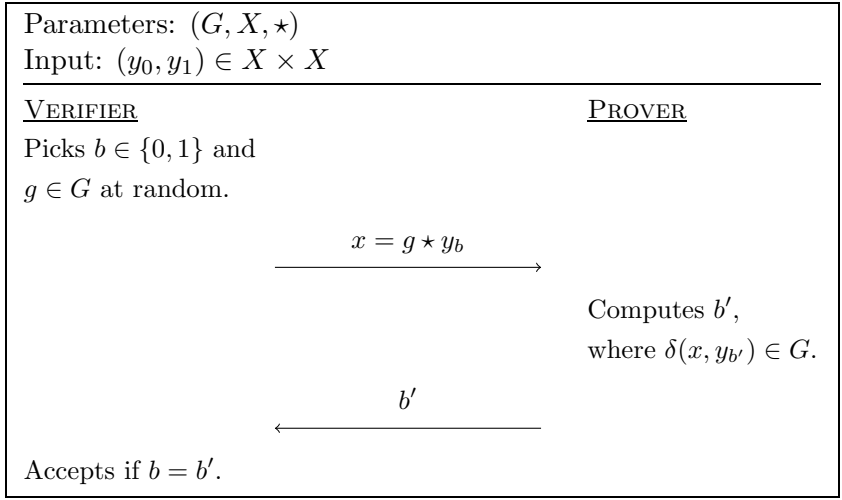

	
	\fbox{
		\procedure[colspace=0cm]{
			Parameters: $(G,X,\star)$\\
			Input: $(y_0,y_1)\in X \times X$}{
			\underline{\textsc{Verifier}} \< \< \underline{\textsc{Prover}}  \\
			\text{Picks $b\in\{0,1\}$ and} \< \< \\
			\text{$g\in G$ at random.} \< \< \\
			\< \sendmessageright*{x=g\star y_b} \< \\
			\< \< \text{Computes $b'$,} \< \< \\
			\<  \< \text{where $\delta(x,y_{b'})\in G$.}  \\
			\< \sendmessageleft*{b'} \< \\
			\text{Accepts if $b=b'$.} \< \< 
		}
	}
	\caption{Interactive protocol for $\overline{\dGAIP}$}
	\label{Fig:dGAIP}
\end{figure}

Using Theorem \ref{Th:collapse}, we can then state the following result.

\begin{Cor}\label{Cor:dGAIP}
	Let $(G,X,\star)$ be an effective non-transitive group action. If $\dGAIP$ is $\NP$-complete, the Polynomial Hierarchy collapses at the second level.
\end{Cor}

Another clue that $\dGAIP$ is not $\NP$-complete comes from the fact that it is in $\L_2^P$, in fact, since $\dGAIP$ is in both $\NP$ and $\coAM$, using Theorem \ref{Th:collapse}, we have the following result.

\begin{Cor}\label{Cor:dGAIPlow}
	Let $(G,X,\star)$ be an effective non-transitive group action. Then $\dGAIP$ is in $\L_2^P$, i.e. the second class of the low hierarchy.
\end{Cor}

We point out that Corollary \ref{Cor:dGAIPlow} has the same consequences of Proposition \ref{Prop:dGAIPcoAM}: in case $\dGAIP$ is $\NP$-complete we have $\PH=\Sigma_{2}^P$ using a result from \cite{schoning1988graph}.

\section{Conclusions}
\label{Sect:concl}

	We showed how some problems deriving from standard computation assumptions for group actions cannot be $\NP$-hard under widely believed conjectures of Complexity Theory. This should not be seen as a completely negative result and their use in cryptography could still be secure: these results concern the worst-case scenario and do not influence the validity of security assumptions used in cryptographic proofs. The fact that a problem is $\NP$-hard is not a certificate of security: for instance, we can think about the Merkle–Hellman knapsack cryptosystem \cite{merkle1978hiding} based on the $\NP$-complete \emph{Subset Sum Problem} but broken in polynomial time \cite{shamir1982polynomial}.\\
	Another consideration about implications of Theorem \ref{Th:problemsRsr} is that $\GAIP$, $\mGAIP$ and $\pGAIP$ are solvable with the same effort both on hard and random instances. This implies that the hardness of these problems in the worst-case equals the hardness in the average-case.\\
	The strong link seen in Section \ref{Subsect:non-tr} between $\dGAIP$ and other central problems in Complexity Theory, like the Graph Isomorphism Problem, suggests that problems from group actions could be situated in the gap between the set of $\NP$-hard ones and problems solvable in polynomial time.

	\section*{Acknowledgments}
	The author acknowledges support from TIM S.p.A. through the PhD scholarship, and Francesco Stocco and Andrea Gangemi for the helpful comments concerning this work.

\bibliographystyle{spmpsci}
\bibliography{main}
\end{document}